\documentclass[10pt, conference, letterpaper]{IEEEtran}
\IEEEoverridecommandlockouts

\ifCLASSINFOpdf
\else
\fi

\usepackage{cite}
\usepackage{amsmath,amssymb,amsfonts}
\usepackage[noend,linesnumbered,algoruled,boxed,lined]{algorithm2e}
\usepackage{graphicx}
\usepackage{textcomp}
\usepackage{xcolor}
\usepackage{multirow}
\usepackage{bm}
\usepackage{booktabs}
\usepackage{tabularx}
\usepackage[caption=false]{subfig}
\DeclareSubrefFormat{parens}{#1(#2)}
\usepackage{ulem}
\usepackage{soul}

\newtheorem{definition}{\noindent \textbf{Definition}}
\newtheorem{lemma}{\noindent \textbf{Lemma}}

\newenvironment{proof}{{\noindent\it Proof. }}{\hfill $\blacksquare$\par}

\allowdisplaybreaks[4]

\begin{document}

\title{Two-Way Aerial Secure Communications via Distributed Collaborative Beamforming under Eavesdropper Collusion}

\author{\IEEEauthorblockN{ Jiahui Li\IEEEauthorrefmark{2}, Geng Sun\IEEEauthorrefmark{2}$^*$, Qingqing Wu\IEEEauthorrefmark{3}, Shuang Liang\IEEEauthorrefmark{4}$^*$, Pengfei Wang\IEEEauthorrefmark{5}, Dusit Niyato\IEEEauthorrefmark{6}}
	
\IEEEauthorblockA{\IEEEauthorrefmark{2}{College of Computer Science and Technology, Jilin University, Changchun 130012, China } \\
\IEEEauthorrefmark{3}{Department of Electronic Engineering, Shanghai Jiao Tong University, Shanghai 200240, China}\\
\IEEEauthorrefmark{4}{School of Information Science and Technology, Northeast Normal University, Changchun 130024, China}\\
\IEEEauthorrefmark{5}{School of Computer Science and Technology, Dalian University of Technology, Dalian 116024, China}\\
\IEEEauthorrefmark{6}{School of Computer Science and Engineering, Nanyang Technological University, Singapore 639798, Singapore}\\
E-mails: lijiahui0803@foxmail.com, sungeng@jlu.edu.cn, qingqingwu@sjtu.edu.cn, \\
liangshuang@nenu.edu.cn, wangpf@dlut.edu.cn, dniyato@ntu.edu.sg \\
\IEEEauthorrefmark{1}{Corresponding authors: Geng Sun and Shuang Liang}}
}

\IEEEtitleabstractindextext{%
\begin{abstract}
Unmanned aerial vehicles (UAVs)-enabled aerial communication provides a flexible, reliable, and cost-effective solution for a range of wireless applications. However, due to the high line-of-sight (LoS) probability, aerial communications between UAVs are vulnerable to eavesdropping attacks, particularly when multiple eavesdroppers collude. In this work, we aim to introduce distributed collaborative beamforming (DCB) into UAV swarms and handle the eavesdropper collusion by controlling the corresponding signal distributions. Specifically, we consider a two-way DCB-enabled aerial communication between two UAV swarms and construct these swarms as two UAV virtual antenna arrays. Then, we minimize the two-way known secrecy capacity and the maximum sidelobe level to avoid information leakage from the known and unknown eavesdroppers, respectively. Simultaneously, we also minimize the energy consumption of UAVs for constructing virtual antenna arrays. Due to the conflicting relationships between secure performance and energy efficiency, we consider these objectives as a multi-objective optimization problem. Following this, we propose an enhanced multi-objective swarm intelligence algorithm via the characterized properties of the problem. Simulation results show that our proposed algorithm can obtain a set of informative solutions and outperform other state-of-the-art baseline algorithms. Experimental tests demonstrate that our method can be deployed in limited computing power platforms of UAVs and is beneficial for saving computational resources.
\end{abstract}

}

\maketitle
% \fi

\IEEEdisplaynontitleabstractindextext
\IEEEpeerreviewmaketitle

% \IEEEraisesectionheading{

%
%Introduction
%
\section{Introduction} % (fold)
\label{sec:introduction}

\par Unmanned aerial vehicles (UAVs), also referred to as drones, can be equipped with communication capabilities to provide a wide range of wireless services~\cite{Dai2022}. Due to the properties of high line-of-sight (LoS) probability, flexibility, and low cost, UAVs are of considerable importance in upcoming generation communications and networks~\cite{Zhang2023}. For instance, in some emergency scenarios, UAVs are potential to be quickly deployed as flying relays to assist intermittent terrestrial networks, facilitating coordination and streamlining rescue efforts~\cite{Prasad2023}. Additionally, UAVs are able to harvest data from remote or hard-to-reach monitoring areas and then forward these data to data fusion centers, accomplishing a flexible data collection~\cite{Hydher2022}. Furthermore, UAVs can act as flying base stations to provide ground users with coverage, enabling on-demand and cost-effective network support~\cite{Abubakar2023}.  

\par Among various applications, UAV-enabled air-to-air (A2A) communication is a promising platform in terrestrial network inadequate scenarios, \textit{e.g.,} emergency assistance, wildlife monitoring, and military surveillance~\cite{Ma2020}. Since the UAV transmitter and receiver are both at high altitudes, the multi-path effect caused by terrain is slight. Thus, UAV-enabled A2A communications are often with increased stability and throughput~\cite{Shi2022}. However, due to the high LoS probability, A2A communications are vulnerable to eavesdropping attacks, particularly when multiple eavesdroppers are colluding~\cite{Cao2023}. In this case, multiple attackers cannot be perfectly detected and they can work together to hide their activities, which renders it challenging for aerial communication systems against such attacks. In general, upper-layer encryptions are well-known approaches to wireless communication links and these approaches provide confidentiality against eavesdroppers. Nevertheless, it is formidable for hardware-limited UAVs to implement such complex encryption algorithms that require significant computational resources. 

\par Physical layer security is a bright secure approach, which utilizes physical characteristics of the wireless channel, such as the randomness of the channel fading, to realize secure communications~\cite{Feng2023}. As such, it can achieve secure communication without relying on complex encryption algorithms or authentication protocols. Due to the high mobility of UAVs, many methods have been proposed to use physical layer security against eavesdroppers. However, the existing works mostly depend on the methods of trajectory design (\textit{e.g.},~\cite{Maeng2022,Yapici2021,Yin2022}) and power allocation (\textit{e.g.},~\cite{Zhang2019,Na2022}) of UAVs, which may confront two major issues. \uline{\textit{First}}, the frequent trajectory changes will undoubtedly increase the energy consumption of UAVs, thereby reducing the corresponding UAV service time. Likewise, power allocation actively restricts the transmit power of UAVs, resulting in overlong transmission time and increased hovering energy costs. \uline{\textit{Second}}, it is complicated for these two methods to handle eavesdropper collusion. This is because the incorrect handling of any eavesdropper will lead to the failure of the total security schemes. Thus, it is desirable to investigate a novel physical layer security to avoid eavesdropper collusion in A2A communications. 

\par In this work, we propose to use distributed collaborative beamforming (DCB)~\cite{Sun2021a} to achieve an efficient secure scheme against eavesdropper collusion. Consider a typical two-way aerial communication between two UAV swarms, we construct each UAV swarm as a UAV virtual antenna array (UVAA). When the system suffers severe secure threats of eavesdropper collusion, we can carefully design the signal distributions (\textit{i.e.}, beam pattern) of UVAAs to suppress the signal strengths toward each detected eavesdropper simultaneously. Meanwhile, the overflowing needless signals of UVAA can be controlled against potential hidden eavesdroppers. As such, by precise signal control, signals toward all eavesdroppers are suppressed, thereby effectively solving the eavesdropper collusion issues.

\par Nevertheless, it is not straightforward to achieve qualified beam patterns for UVAAs. On the one hand, the beam patterns of UVAAs are affected by the positions and transmit powers~\cite{Jayaprakasam2017} of UAVs, which need to be attentively determined. On the other hand, during UAVs fine-tuning their positions, the energy cost of UAVs will no doubt increase. Thus, we should well capture and balance the trade-offs between secure performance and energy efficiencies of UAVs. Accordingly, we aim to control the decision variables above and consider a multi-objective optimization method in DCB-enabled A2A communications under eavesdropper collusion. Our main contributions are listed as follows. 

\begin{itemize}

  \item \textit{Novel Paradigm for Solving Multi-eavesdropper Collusion:} We consider a typical DCB-enabled aerial two-way communication of two UAV swarms under eavesdropper collusion. Specifically, eavesdroppers collude based on signal detection, which leads to the worst wiretap case. In this case, we introduce DCB into each UAV swarm and use the signal processing method to handle eavesdroppers. To the best of our knowledge, this is the first work that uses DCB against eavesdropper collusion. 
  % \hl{maximum ratio combining (MRC)}

  \item \textit{NP-hard Optimization Problem Formulation:} We aim to maximize the two-way secrecy rate of the aerial communication and suppress the needless overflowing signals to increase the secure performance of the system. We also minimize the energy costs of the UAVs to ensure the energy-efficiency. Due to their conflicting relationships, we formulate a multi-objective optimization problem (MOP) to optimize these goals simultaneously. Then, we prove that it is an NP-hard problem.

  \item \textit{Enhanced Multi-objective Swarm Intelligence Method:} Due to the NP-hardness, the optimization problem is difficult to be solved. Thus, we propose a novel enhanced multi-objective swarm intelligence method to solve the problem. The proposed algorithm is enhanced by the properties of the formulated MOP, and it is able to find candidate solutions to the NP-hard problem with low computational complexity. 
  
  \item \textit{Simulations and Experiments:} Simulation results show that our proposed method outperforms various benchmarks and state-of-the-art algorithms, and also is robust under phase synchronization error, channel state information (CSI) error, and UAV jitter. Moreover, experimental tests demonstrate that our method can be deployed in limited computing power platforms of UAVs and is beneficial to saving computational resources.

\end{itemize}

\par The rest of this paper is arranged as follows. Section \ref{sec:related_work} reviews the related works. Section \ref{sec:models_and_preliminaries} presents the models and preliminaries. Section \ref{sec:problem_formulation_and_analysis} formulates the MOP. Section \ref{sec:algorithm} proposes the algorithm. Section \ref{sec:simulation_results_and_analysis} shows the simulation results and Section \ref{sec:conclusion} concludes the paper. 

% section introduction (end)

%
%Related works
%
\section{Related Work} % (fold)
\label{sec:related_work}

\par In this section, we briefly introduce the related works of A2A communications, physical layer security, and DCB in UAV networks to highlight our innovations and contributions. 
% we aim to use DCB to achieve physical layer security in the A2A communications. Thus,

%
% A2A Communication in UAV Networks
%
\vspace{+0.5mm}
\noindent \textbf{A2A Communication-enabled by UAVs.} Due to the high-LoS probability, UAV-enabled A2A communications achieved significant stability and throughput. Ma \textit{et al.}~\cite{Ma2020} used a three-dimensional (3D) Markov mobility model to characterize the movements of the UAV for the systematic evaluation of UAV-enabled A2A communications. Other works related to A2A communications also contain~\cite{Li2020,Humadi2021,Fakhreddine2022} and their cited previous literature. However, A2A communications suffer from severe security problems, especially in eavesdropper collusion cases. None of these existing works consider proposing a physical layer security scheme to solve this issue. 
% Moreover, Li \textit{et al.}~\cite{Li2020} integrated the coordinated multi-point transmission technique with stochastic geometry theory, and thereby proposed a novel 3D cellular model consisting of aerial base stations and aerial user equipment.

%
% Physical Layer Security in UAV Networks
%
\vspace{+0.5mm}
\noindent \textbf{Physical Layer Security in UAV Networks.} On the one hand, several existing works considered power allocation-based physical layer security methods in UAV networks~\cite{Yapici2021,Yin2022}. Maeng \textit{et al.}~\cite{Maeng2022} proposed a linear precoder design for UAVs and derived optimal power splitting factor. On the other hand, trajectory design is another effective approach to handle eavesdroppers~\cite{Zhang2019}. Nevertheless, finding a suitable power or trajectory to avoid all eavesdroppers is almost impossible. Thus, this motivates us to investigate a novel physical layer security to solve the issue of eavesdropper collusion. 
% In general, physical layer security consists of two main types of methods. However, power allocation also may reduce the transmit power of UAVs, resulting in overlong transmission time and increased hovering energy costs. Nevertheless, the frequent trajectory design will undoubtedly increase the energy consumption of UAVs, which may reduce the service time of the UAV.  the frequent trajectory design will undoubtedly increase the energy consumption of UAVs, which may reduce the service time of the UAV. Moreover, For instance, the authors in established specific channels for users and eavesdroppers by jointly optimizing the UAV's trajectory and the transmit power of the legitimate transmitter.

%
% DCB in UAV Networks
%
\vspace{+0.5mm}
\noindent \textbf{DCB in UAV Networks.} Boosted by the mobility of UAVs, DCB can achieve promising prosperities to improve the transmission abilities of UAV communication systems. For example, the authors in~\cite{Mohanti2022} proposed a complete algorithmic framework or system implementation of DCB on a UAV swarm. References~\cite{Sun2021,Shi2022,Li2021} used UAV-enabled DCB to achieve long-range transmission, two-way communications, and physical layer security, respectively. However, these works overlook the eavesdropper collusion issues. 
%  Additionally, some efforts have been put into validating the real-world feasibility of UAV-based DCB. Mozaffari \textit{et al.}~\cite{Mozaffari2019} constructed a UAV-enabled linear antenna array to minimize the wireless transmission time of UAVs. Moreover,

\par Different from the existing works, we aim to use DCB against the worst wiretap case of eavesdropper collusion. This is challenging since we need to balance UAVs' secure performance and energy efficiency, and control complex decision variables. In what follows, we will present the models and preliminaries of the considered system, thereby characterizing the relationships between the decision variables with secure performance and energy efficiency of the system.

%
%Models and preliminaries
%
\section{Models and Preliminaries} % (fold)
\label{sec:models_and_preliminaries}

\par In this section, we first present the system overview. Then, we detail the considered models, including the virtual antenna array, A2A transmission, MRC eavesdropping, and UAV energy cost models, to characterize the objectives and decision variables. The main notations are presented in Table \ref{tab:notations}.

% \usepackage{booktabs}
\begin{table}[]
\centering
\caption{Main notations}
\label{tab:notations}
\begin{tabularx}{3.5in}{p{1.5cm}p{6.5cm}}
\toprule
\textbf{Symbol}         & \textbf{Definition}                                                     \\ \midrule
$\mathcal{U}_1$, $\mathcal{U}_2$ & The sets of the rotary-wing UAVs in the first and second swarms, respectively \\
$\mathcal{KE}$, $\mathcal{UE}$ & The sets of the known and unknown eavesdroppers, respectively \\
$AF_i$, $G^{A2A}_i$ & The array factor and gain of the $i$th UVAA, respectively \\
$P^{\text{LoS}}_{i,k}$ & The probability of $i$th UVAA and eavesdropper $k$\\
$C_{KE}$ & The minimum two-way known secrecy capacity of the system\\
$E_{i,j}$ & The energy cost of $j$th UAV of $i$th UVAA for performing the communication\\
$\boldsymbol{P}$, $\boldsymbol{\Omega}$ & The positions and excitation current weights of the UAVs of two UVAAs, respectively \\
$\boldsymbol{u}$ & The select UAV receivers of two UVAAs \\
$\boldsymbol{X}$ & The solution to the formulated MOP \\
$\mathcal{P}$, $\mathcal{A}$ & The population and archive of the proposed MOALO-RSI algorithm, respectively \\
\bottomrule
\end{tabularx}
\end{table}

% Subsection
% System Overview
%
\vspace{+0.5mm}
\noindent \textbf{System Overview.} As shown in Fig.~\ref{fig:network-model}, we consider a two-way aerial communication system under eavesdropper collusion. As can be seen, the system has two UAV swarms denoted as $\{ \mathcal{U}_i | i \in 1, 2\}$ deployed in a network inadequate area for emergency assistance, wildlife monitoring, military surveillance, etc. Note that we can easily extend two UAV swarms system to multiple by introducing routing and networking protocols. In this work, we only consider two UAV swarms for the sake of simplicity and easy access to insights. Accordingly, a set of randomly distributed UAVs denoted as $\mathcal{U}_1 \triangleq \{j | 1, 2, \dots, N_{U}\}$ are in a limited area denoted as $A_{U1}$. Moreover, another UAV swarm denoted as $\mathcal{U}_2 \triangleq \{j | 1, 2, \dots, N_{U}\}$ is dispatched in area $A_{U2}$ which is far and no overlapping from $A_{U1}$. In both the areas and the surrounding ground, there are several eavesdroppers denoted as $\mathcal{E} \triangleq \{k | 1, 2, \dots, N_{E}\}$ that are randomly placed. We assume that some eavesdroppers can be detected~\cite{Zhang2019} while others are undetectable but are potential threats. Accordingly, the eavesdroppers in $\mathcal{E}$ can be divided into known eavesdroppers and unknown eavesdroppers as $\mathcal{KE}$ and $\mathcal{UE}$, respectively. Moreover, we assume that each UAV of $\mathcal{U}_1$ and $\mathcal{U}_2$ is equipped with a single omni-directional antenna, and all the eavesdroppers collude.
% in an MRC manner~\cite{}
 
\par The system works as follows. Consider a particular case that the UAVs in $\mathcal{U}_1$ and $\mathcal{U}_2$ need to exchange emergency information and data. Due to the lack of terrestrial access points, these two UAV swarms will construct two UVAAs (\textit{i.e.}, UVAA 1 and UVAA 2) simultaneously. We assume that the UAVs within the same virtual antenna array are synchronized regarding the carrier frequency, time, and initial phase by using the methods in~\cite{Mohanti2022}. Moreover, the data sharing among UAVs of each virtual antenna array is achieved by using the method in~\cite{Feng2013}. In addition, we assume that the UAVs obtain the quantization version of CSI via method in~\cite{Ahmad2022}. Following this, they will select a suitable receiver from a different UAV swarm as the receiver and then achieve a two-way aerial communication. Without loss of generality, we consider a 3D Cartesian coordinate system. As such, the positions of the $j$th UAV in $i$th swarm $\mathcal{U}_i$ and the $k$th eavesdropper are indicated as $(x^{U}_{i,j}, y^{U}_{i,j}, z^{U}_{i,j})$ and $(x^{E}_{k}, y^{E}_{k}, 0)$, respectively.

% Figure
% System_Model
%
\begin{figure}
  \centering
  \includegraphics[width=3.5in]{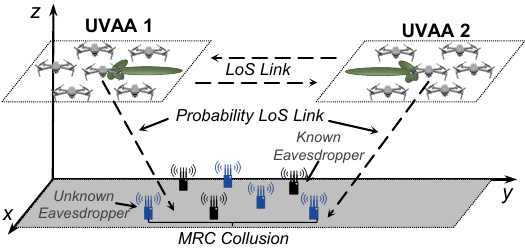}
  \caption{A two-way DCB-enabled aerial communication system under the known and unknown eavesdropper collusion.}
  \label{fig:network-model}
\end{figure}

% Subsection
% Virtual Antenna Array Model
%
\vspace{+0.5mm}
\noindent \textbf{Virtual Antenna Array Model.} The electromagnetic waves emitted by UAV antennas of a UVAA will be superposed and offset, thereby producing a beam pattern with a sharp mainlobe and low sidelobes. Mathematically, we use array factor to measure the signal strengths in different directions. Let $(x^{U}_{i,j}, y^{U}_{i,j}, z^{U}_{i,j})$ and $\omega_{i,j}$ be the \textit{3D coordinates} and \textit{excitation current weight} of $j$th UAV in $\mathcal{U}_i$. Then, the array factor of $i$th UVAA is given by
\begin{equation}
   \label{eq:AF}
   \begin{aligned}
      A&F_i(\theta, \phi)= \\
      &\sum \limits _{j = 1}^{N_{U}} \omega_{i,j} e^{j^{u} \left [ \frac{2\pi}{\lambda} \left ({{{ x^{U}_{i,j} \sin \theta \cos \phi + {y^{U}_{i,j}}\sin \theta \sin \phi + {z^{U}_{i,j}}}\cos \theta } }\right) + \Phi_{i,j}\right]},
   \end{aligned}
\end{equation}
\noindent where $\theta$ and $\phi$ which range $[0, \pi]$ and $[-\pi, \pi]$ are the elevation and azimuth angles from the center of $\mathcal{U}_i$ ($x^c_i$, $y^c_i$, $z^c_i$) to any receiver, respectively. Other parameters shown in Eq.~\eqref{eq:AF} are related to communications, \textit{e.g.}, $j^{u}$ is the imaginary unit and $\Phi_{i,j}$ is the initial phase of $j$th array element in $i$th UVAA.

% Subsection
% A2A transmission model
%
\vspace{+0.5mm}
\noindent \textbf{A2A Transmission Model.} Due to the high altitudes of UAVs and usage of DCB, the A2A transmission should follow an LoS channel condition~\cite{Mozaffari2019}. Let $d^{A2A}_i$ denote the distance between the center of $i$th UVAA and the corresponding receiver, the transmission rate is given by
\begin{equation}
  \label{eq:R_A2A}
  \begin{aligned}
    R^{A2A}_i=B \log _{2}(1+\frac{P^{t}_i K_{0} G_{i}^{A2A} {d^{A2A}_i}^{-\alpha}}{\sigma^{2}}).
  \end{aligned}
\end{equation}
\noindent Other parameters in Eq.~\eqref{eq:R_A2A} are related to communications. Specifically, $B$ is the bandwidth, $P^t_i$ is the total transmission power of the $i$th UVAA, $K_{0}$ is the constant path loss coefficient, and $\sigma^{2}$ is the noise power. Another key parameter in Eq.~\eqref{eq:R_A2A} is the antenna gain $G_{i}^{A2A}$. Let $(\theta_i, \phi_i)$ denote the direction from the center of $i$th UVAA to the receiver, $G_{i}^{A2A}$ is expressed as 
\begin{equation}
  \label{eq:gain}
  \begin{aligned}
    G_{i}^{A2A}=\frac{4\pi\left|AF_i\left(\theta_{i}, \phi_{i}\right)\right|^{2} w\left(\theta_{i}, \phi_{i}\right)^{2}}{\int_{0}^{2 \pi} \int_{0}^{\pi}|AF_i(\theta, \phi)|^{2} w(\theta, \phi)^{2} \sin \theta \mathrm{d} \theta \mathrm{d} \phi} \eta ,
  \end{aligned}
\end{equation}
\noindent where $w{(\theta,\phi)}$ is the magnitude of the far-field beam pattern of each UAV, $\eta \in [0, 1]$ is the antenna array efficiency~\cite{Mozaffari2019}. 
% where $d^{A2A}_i= \sqrt{ (x^c_i-x^r_i)^2+(y^c_i-y^r_i)^2+(z^c_i-z^r_i)^2 }$, in which ($x^r_i$, $y^r_i$, $z^r_i$) is the location of the receiver of $i$th UVAA.

% Subsection
% MRC eavesdropping model
%
\vspace{+0.5mm}
\noindent \textbf{Eavesdropper Collusion Model.} The ground-placed eavesdroppers suffer a probability LoS channel condition. Let $\theta_{i, k}$ denote the elevation angle between the center of $i$th UVAA and eavesdropper $k$, then the LoS probability is~\cite{Faraci2019}
\begin{equation}
  \label{eq:LoS_P}
  \begin{aligned}
    P^{\mathrm{LoS}}_{i,k}=\left(1+b_1 e^{-b_2(\theta_{i, k}-b_1)} \right)^{-1},
  \end{aligned}
\end{equation}
\noindent where $b_{1}$ and $b_{2}$ are constant values determined by environment. Then, the corresponding non-line-of-sight (NLoS) probability is given by $P^{\mathrm{NLoS}}_{i,k}=1-P^{\mathrm{LoS}}_{i,k}$. Accordingly, let $d_{i,k}$ and $G_{i,k}$ denote the distance and antenna gain between the center of $i$th UVAA and the eavesdropper $k$, respectively, the corresponding SNR is expressed as 
\begin{equation}
  \begin{aligned}
    \gamma_{i,k}= \frac{P^t_i K_0 G_{i,k} d_{i,k}^{-\alpha}[P^{\mathrm{LoS}}_{i,k} \mu_{\mathrm{LoS}}+P^{\mathrm{NLoS}}_{i,k} \mu_{\mathrm{NLoS}}]^{-1}} {\sigma^2},
  \end{aligned}
\end{equation}
\noindent where $\mu_{\mathrm{LoS}}$ and $\mu_{\mathrm{NLoS}}$ are the attenuation factors of LoS and NLoS links, respectively.
% $d_{i,k}= \sqrt{ (x^c_i-x^E_k)^2+(y^c_i-y^E_k)^2+(z^c_i-z^E_k)^2 }$, and

\par The eavesdroppers adopt maximum ratio combining (MRC) based on signal detection~\cite{Dung2021}. This is due to the fact that MRC collusion leads to the worst wiretap case and is easy to extend for other eavesdropper collusion conditions. Specifically, each eavesdropper can be regarded as an antenna in a multi-antenna system. Then, the system weights and combines the received signals from colluded eavesdroppers via MRC diversity-combining technique~\cite{Dung2021}. In this case, after assigning the best-weighted value of different branches, the combined output SNR from the $i$th UVAA is 
\begin{equation}
  \label{eq:snr_all}
  \begin{aligned}
    {\gamma_{\Sigma}}_i= \sum_{k=1}^{N_{E}} \gamma_{i,k}.
  \end{aligned}
\end{equation}

\noindent As such, the achievable rate of colluding eavesdroppers from $i$th UVAA is given by $R_{i}^{E}= B \log_2 (1+ {\gamma_{\Sigma}}_i)$. Using $R^{A2A}_i$ and $R_{i}^{E}$, we can express the \textit{minimum two-way achievable secrecy capacity} of the two-way communication as
\begin{equation}
  \label{eq:ce}
  \begin{aligned}
    C_E=\min_{i \in \{1, 2\}, k \in \mathcal{E}} \{R^{A2A}_i- R_{i}^{E}\}.
  \end{aligned}
\end{equation}

\noindent As can be seen, our defined secure transmission performance considers the minimum capacity between two-way communications, which means that improving $C_E$ can enhance the secure performance in both UVAA 1 to UVAA 2 and UVAA 2 to UVAA 1. 

\par However, when we formulate the optimization problem, only the detailed information of known eavesdroppers is feasible. Thus, we also define the \textit{minimum two-way known secrecy capacity} as follows:
\begin{equation}
  \label{eq:cke}
  \begin{aligned}
    C_{KE}=\min_{i \in \{1, 2\}, k \in \mathcal{KE}} \{R^{A2A}_i- R_{i}^{E}\}.
  \end{aligned}
\end{equation}

\noindent It can be seen from Eqs.~\eqref{eq:LoS_P}-\eqref{eq:cke} that both the minimum two-way achievable secrecy capacity and minimum two-way known secrecy capacity are determined by the array factors of the UVAAs. In other words, we can carefully control the signal distributions of UVAAs to avoid these eavesdroppers.

% Subsection
% UAV Energy Cost Model
%
\vspace{+0.5mm}
\noindent \textbf{UAV Energy Cost Model.} We adopt the typical rotary-wing UAV, where the propulsion energy cost is the main component of its total energy cost while other components can be negligible~\cite{Sun2021}. Let $v$ be the fly speed of the UAV, the energy cost of the UAV for flying in the horizontal plane is given by
\begin{equation}
   \label{eq:energy-2d}
   \begin{aligned}
      P(v)=&P_{B}(1+\frac{3v^2}{v_{tip}^2})+P_{I}(\sqrt{1+\frac{v^4}{4v_0^4}}-\frac{v^2}{2v_0^2})^{1/2}+\\&\frac{1}{2}d_0\rho sAv^3,
   \end{aligned}
\end{equation}
\noindent where $P_{B}$, $P_{I}$, $v_{tip}$, $v_{0}$, $d_{0}$, $s$, $\rho$, and $A$ are blade profile constant, induced powers constant, tip speed of the rotor blade, mean rotor induced velocity in hovering, fuselage drag ratio, rotor solidity, air density, and rotor disc area, respectively.

\par According to Eq.~\eqref{eq:energy-2d}, the propulsion energy cost of UAV can be extended into the 3D form, \textit{i.e.}, 
\begin{equation}
  \label{eq:energy-3d}
   \begin{aligned}
     E(T) \approx &\int_0^TP( v(t))dt+{\frac12m_{D}(v{(T)}^2- v{(0)}^2)}+\\&{m_{D}g(h(T)-h(0))},
   \end{aligned}
\end{equation}
\noindent where $v(t)$ is the instantaneous drone speed at time $t$, $T$ represents the end time of the flight, $m_{D}$ is the aircraft mass of a UAV, and $g$ is the gravitational acceleration.

% Section
% Problem Formulation and Analysis
%
\section{Problem Formulation and Analysis} % (fold)
\label{sec:problem_formulation_and_analysis}

\par In this section, we formulate a secure and energy-efficient two-way aerial communication problem. First, we state the main idea of the optimization problem. Then, the decision variables and optimization objectives are presented. Finally, we construct an MOP and prove that it is an NP-hard problem.

% Subsection
% Problem Statement:
%
\vspace{+0.5mm}
\noindent \textbf{Problem Statement.} The considered system concerns two goals, \textit{i.e.}, improving the achievable secrecy capacity and reducing the energy cost for fine-tuning the positions of UAVs. 

\par \textit{(i)} Since we cannot obtain the exact information of unknown eavesdroppers, we cannot access the minimum two-way achievable secrecy capacity shown in Eq.~\eqref{eq:ce}. Thus, we aim to take two measures simultaneously to optimize the achievable secrecy capacity. \uline{\textit{First}}, we aim to maximize the known secrecy capacity shown in Eq.~\eqref{eq:cke} to reduce the signal qualities obtained by known eavesdroppers. \uline{\textit{Second}}, we can minimize all the signals except the target direction, thereby avoiding potential unknown eavesdroppers. These two measures can be achieved by controlling the array factor (which is determined by the \textit{\uline{3D positions}} and \uline{\textit{excitation current weights}}) and \uline{\textit{aerial receiver}} selection. 

\par \textit{(ii)} We need to fine-tune UAVs' positions to enhance DCB performance, which will result in extra energy costs. Thus, the \uline{\textit{position changes}} of UAVs should be minimized by considering the energy efficiency.

% Subsection
% Decision Variables
%
\vspace{+0.5mm}
\noindent \textbf{Decision Variables.} Based on above analyses, the following decision variables need to be determined jointly: \textit{(i)} $\boldsymbol{P} = \left\{(x_{i,j}^U, y_{i,j}^U, z_{i,j}^U) | i \in \{ 1, 2 \}, j \in \mathcal{U}_i \right\}$, a matrix consisting of continuous variables denotes the 3D positions of UAVs in UVAA1 and UVAA2 for performing DCB. \textit{(ii)} $\boldsymbol{\Omega} = \left\{ \omega_{i,j} | i \in \{ 1, 2 \}, j \in \mathcal{U}_i \right\} $, a matrix consisting of continuous variables denotes the excitation current weights of UAVs in UVAA1 and UVAA2 for performing DCB. \textit{(iii)} $\boldsymbol{u} = \left\{ u_i | i \in \{1, 2\} , u_i \in \mathcal{U}_i \right\} $, a vector consisting of discrete variables represents the IDs of selected aerial receivers of UVAA 1 and UVAA 2. 
% \begin{figure*}
%   \begin{equation}\label{eq:decision variable}
%     \begin{aligned}
%       \boldsymbol{X} = & [\boldsymbol{P}, \boldsymbol{\Omega}, \boldsymbol{u}] = &
%       \begin{bmatrix}
%         \begin{smallmatrix}  
%           x_{1, 1}^{U}, & \cdots, & x_{1, N_{U}}^{U}, & y_{1, 1}^{U}, & \cdots, & y_{1, N_{U}}^{U}, & z_{1, 1}^{U}, & \cdots, & z_{1, N_{U}}^{U}, &\omega_{1, 1}, & \cdots, & \omega_{1, N_{{UAV}}}, & u_{1} \quad \\
%           \undermat{\boldsymbol{P}}{x_{2, 1}^{U}, & \cdots, & x_{2, N_{U}}^{U}, & y_{2, 1}^{U}, & \cdots, & y_{2, N_{U}}^{U}}, & z_{2, 1}^{U}, & \cdots, & z_{2, N_{U}}^{U}, &\undermat{\boldsymbol{\Omega}}{\omega_{2, 1}, & \cdots, }& \omega_{2, N_{{UAV}}}  \undermat{\boldsymbol{u}}{,} & u_{2} \quad
%         \end{smallmatrix}
%       \end{bmatrix}.\\
%       \vspace{+1 mm}
%     \end{aligned}
%   \end{equation}
% \hrulefill
% \end{figure*}

% Subsection
% Decision Variables
%
\vspace{+0.5mm}
\noindent \textbf{Optimization Objectives.} To achieve DCB-enabled secure and energy-efficient two-way communication, we consider the following objectives simultaneously.

\par \textit{Optimization Objective 1}: Our first objective is to maximize the minimum two-way known secrecy capacity of the system for avoiding known eavesdroppers. To this end, we need to jointly optimize the $\boldsymbol{P}$, $\boldsymbol{\Omega}$, and $\boldsymbol{u}$. Then, the first objective is given by 
\begin{equation}
  \label{eq: objecitve1}
  \begin{aligned}
    f_{1}(\boldsymbol{P}, \boldsymbol{\Omega}, \boldsymbol{u})= C_{KE}.
  \end{aligned}
\end{equation}

\par \textit{Optimization Objective 2}: Our second objective is to minimize the signal densities of the UVAAs in all directions to avoid unknown eavesdroppers. Specifically, we adopt the \textit{maximum sidelobe level (SLL)} for evaluating the signal strength except for the targeted direction of a UVAA. Let $(\boldsymbol{\theta}^{SLL}_i, \boldsymbol{\phi}^{SLL}_i)$ denote the direction set except targeted direction, the maximum SLL of $i$th UVAA is expressed as 
\begin{equation}
  \begin{aligned}
    f_{SLL_i} = \frac{ \max|AF_i(\boldsymbol{\theta}^{SLL}_i, \boldsymbol{\phi}^{SLL}_i)|}{AF_i(\theta_i, \phi_i)}.
  \end{aligned}
\end{equation}
\noindent Then, we minimize the maximum term among the maximum SLLs of UVAAs. To this end, we jointly optimize $\boldsymbol{P}$ and $\boldsymbol{\Omega}$. As such, the second objective is given by
\begin{equation}
  \label{eq: objecitve2}
  \begin{aligned}
    f_{2}(\boldsymbol{P}, \boldsymbol{\Omega})= \max_{i \in \{1, 2\}} \{ f_{SLL_i} \}. 
  \end{aligned}
\end{equation}

\noindent Accordingly, we can minimize Eq.~\eqref{eq: objecitve2} to optimize the maximum SLLs of UVAA 1 and UVAA 2 simultaneously. 

\par \textit{Optimization Objective 3}: Our third objective is to minimize the energy costs of all UAVs for fine-tuning their positions. We optimize the $\boldsymbol{P}$ to achieve this goal, and the third objective is given by 
\begin{equation}
  \label{eq: objecitve3}
  \begin{aligned}
    f_{3}(\boldsymbol{P})= \sum_{i \in \{1, 2\}} \sum_{j \in \mathcal{U}_i} E_{i,j},
  \end{aligned}
\end{equation}
\noindent where $E_{i,j}$ is the energy cost of $j$th UAV of $i$th UVAA for moving to the DCB position. Note that we calculate $E_{i,j}$ according to $\boldsymbol{P}$ and the original position of UAVs $\boldsymbol{P}^{r}$ by using the method in~\cite{Sun2021}. Following~\cite{Sun2021}, we assume that the UAVs fly first horizontally and then vertically, as this approach is the easiest to be controlled by automatic flight control.

% Subsection
% Problem Formulation
%
\vspace{+0.5mm}
\noindent \textbf{Problem Formulation.} By considering these conflicting objectives, the optimization problem can be given in an MOP formulation as 
\begin{subequations}
  \label{eq:formulation}
  \begin{align}
    {\underset{\boldsymbol{X} = \{\boldsymbol{P}, \boldsymbol{\Omega}, \boldsymbol{u}\} }{\text{min}}} \ & F=\{-f_{1}, f_{2}, f_{3}\},\\
    \text{s.t.} \quad \quad
    & (x_{i,j}^U, y_{i,j}^U, z_{i,j}^U) \in \mathbb{R}^{3 \times 1}_i , \forall i \in \{1, 2\}, \forall j \in \mathcal{U}_i, \label{eq:const1}\\
    & 0 \leqslant \omega_{i, j} \leqslant  1, \forall i \in \{1, 2\}, \forall j \in \mathcal{U}_i \label{eq:const2},\\
    & u_i \in \mathcal{U}_{\{1,2\}-i}, \forall i \in \{1, 2\} , u_i \in \mathcal{U}_i \label{eq:const3}, \\
    & d_{(j_1, j_2)}^i \geq d_{min}, \forall i \in \{1, 2\}, \forall j_1, j_2 \in \mathcal{U}_i \label{eq:const4},
  \end{align}
\end{subequations}

\noindent where $\boldsymbol{X}$ is the decision variable set of the problem. Then, $\mathbb{R}^{3 \times 1}_i$ is the reachable area of the $i$th UVAA. Moreover, $\omega_{i, j}$ indicates the ranges of the excitation current weights. In particular, $\omega_{i, j}=0$ means the antenna array of the UAV is switched off, while $\omega_{i, j}=1$ indicates that the antenna transmits signals in the maximum transmission power. In addition, constraint \eqref{eq:const3} shows that the UVAA must select a receiver from a different UAV swarm. Furthermore, constraint \eqref{eq:const4} presents a hard condition, \textit{i.e.}, minimum separation distance, to avoid the collision between any UAVs in all UVAAs. 

\par Note that the formulated optimization problem shown in Eq.~\eqref{eq:formulation} can be simplified as a nonlinear multi-dimensional knapsack problem. Thus, the optimization problem is NP-hard. Due to its NP-hardness, the formulated problem cannot be solved in polynomial time. Moreover, since decision variables and optimization objectives have complex and non-linear relationships, it is challenging to give a powerful approximation algorithm or convex optimization method for the problem. Artificial intelligence methods are promising to solve such complex problems efficiently, \textit{e.g.}, reinforcement learning and swarm intelligence algorithm can solve sequence decision-making and static deployment problems, respectively. Considering that our problem is in a static scenario, we aim to propose a novel multi-objective swarm intelligence algorithm to find the candidate solutions to this problem.

% Section
% Algorithm
%
\section{Algorithm Design} % (fold)
\label{sec:algorithm}

\par In this section, we first present the necessary principles of the adopted swarm intelligence algorithm. Then, we enhance the algorithm by the properties of our optimization problem. 

% Subsection
% Framework of population-based heuristic algorithm
%
\subsection{Necessary Principles of Multi-objective Swarm intelligence}
\label{ssec:framework}

\par Swarm intelligence is a branch of artificial intelligence that can potentially solve some NP-hard problems~\cite{Tang2021}. Specifically, swarm intelligence algorithms maintain and improve a set of candidate solutions to the optimization problem. Then, these candidate solutions are improved by the iterations. This work introduces a multi-objective ant-lion optimizer (MOALO)~\cite{Mirjalili2017b}, which is one of the state-of-the-art algorithms, as the solving framework. 
% We first show the main steps of MOALO and then analyze the challenging issues that MOALO needs to be tackled in solving our optimization problem. 

% Subsubsection
% Main Steps of MOALO
%
\vspace{+0.5mm}
\noindent \textbf{Main Steps of MOALO.} The main steps of MOALO are summarized as follows.

\begin{enumerate}
  \item Generate Candidate Solutions: MOALO first generates $N$ candidate solutions, and each solution should be a feasible solution to the optimization problem. As such, we denote the population as $\mathcal{P}=\{\boldsymbol{X}_1, \boldsymbol{X}_2, ..., \boldsymbol{X}_{N}\}$, in which $\boldsymbol{X}_n=[\boldsymbol{P}_n, \boldsymbol{\Omega}_n, \boldsymbol{u}_n]$.
   % (\textit{i.e.}, in the form of Eq.~\eqref{eq:decision variable})

  \item Calculate Objective Functions: The candidate solutions will be evaluated by calculating the objective values by using Eqs.~\eqref{eq:AF}-\eqref{eq:formulation}. Mathematically, we let $F_{n}=[f_{1,n}, f_{2,n}, f_{3,n}]$ denoting the three optimization objective values of the $n$th candidate solutions. 

  \item Compare and Reserve Solutions: MOALO maintains a candidate solution set, namely, archive, to save the elite solutions through the previous iterations, which can be denoted as $\mathcal{A}=\{\boldsymbol{X}_1^A, \boldsymbol{X}_2^A, ...\}$. In each iteration, MOALO combines the current population and archive, \textit{i.e.}, $\mathcal{A} \leftarrow \{\mathcal{A}, \mathcal{P}\}$. Then, the solution with better objective values (\textit{i.e.}, non-dominated solution defined below) will be reserved in $\mathcal{A}$. Due to the nature of MOPs, the comparison between different solutions cannot be done by arithmetic relational operators. Thus, we adopt Pareto dominance to prioritize solutions. 
  % , and $\mathcal{A}'$ is the new archive, \textit{i.e.}, $\mathcal{A} \leftarrow \mathcal{A}'$

  \vspace{0.5 mm}
  \begin{definition} [Pareto dominance]
  $\boldsymbol{X}$ \textit{Pareto dominance} $\boldsymbol{X}'$ iff: $\left[\forall o \in\{1,2,3\},f_o(\boldsymbol{X}) \leq f_o(\boldsymbol{X}')\right] \wedge \left[\exists o \in\{1,2,3\} f_o(\boldsymbol{X}) < f_o(\boldsymbol{X}')\right]$.
  \end{definition}
  
  \par The non-dominated solutions reserved by the archive can be defined as follows. 

  \vspace{0.5 mm}
  \begin{definition} [Non-dominated solution]
  $\boldsymbol{X}$ is called non-dominated solution iff: $\exists\mkern-13mu/ \boldsymbol{X'} \in \mathcal{A}$ Pareto dominance $\boldsymbol{X}$.
  \end{definition}
  \vspace{0.5 mm}

  \par If the archive reaches its maximum capacity, MOALO will use the crowding mechanism to remove the solutions with similar trade-offs from the archive, thereby ensuring MOALO covers as many trade-offs as possible.

  \item Update Candidate Solution: MOALO selects candidate solutions from $\mathcal{A}$ by the roulette wheel selection. Then, MOALO uses the selected candidate solutions to update the candidate solutions as follows:
  \begin{equation}\label{eq:solution-update}
  	\boldsymbol{X}_{n}=(\boldsymbol{X}^{R}+\boldsymbol{X}^{A})/2,
  \end{equation}
  \noindent where $\boldsymbol{X}^{R}$ is the guide solution calculated by heuristic principles of MOALO~\cite{Mirjalili2017b}, and $\boldsymbol{X}^{A}$ is the solution selected from $\mathcal{A}$.

  \item Terminate Algorithm: Determine whether the termination condition is reached. If no, repeat steps 2) -5). If yes, the candidate solutions within $\mathcal{A}$ are the final solutions.
\end{enumerate}

% Subsubsection
% Challenging Issues.
%
\vspace{+0.5 mm}	
\noindent \textbf{Challenging Issues.} MOALO needs to overcome the following issues when solving the formulated problem.

\begin{itemize}
    \item The formulated problem has a large-scale solution space. However, MOALO searches the solution space in a uniformly random manner. In this case, MOALO may fall into local optima.

    \item When solving our MOP, the non-dominated solutions found by MOALO may be not good solutions to our formulated problem. This is because some non-dominated solutions are with extreme trade-offs (\textit{e.g.}, bias toward one objective but overlooking the other two objectives). For instance, we do not need a solution that consumes zero energy but achieve marginal secure performance.

    \item The formulated problem involves integer decision variables. However, MOALO can only deal with continuous decision variables. 
\end{itemize}

\par In the following, we propose several improvements to make MOALO tackle these challenges.

%subsection
%EMOMVO
%
\subsection{Our Proposed MOALO-RSI Method}
\label{ssec:EMOMVO}

\par In this section, we propose an MOALO with random walk-based initialization, sorting-based population evolution, and integer update methods (MOALO-RSI). 

%subsubsection
%Random Walk-based Initialization
%
\vspace{+0.5mm}
\noindent \textbf{Random Walk-based Initialization.} The initial population state of the swarm intelligence algorithm is able to guide the search direction of the following iterations~\cite{Kazimipour2014}. Thus, we can summarize some fundamental principles and add them to the population initialization scheme. 

\begin{lemma}
\label{lemma:energy}
    The space around $\boldsymbol{P}^r$ is more informative since we are more likely to find positions of UAVs with low energy consumption and relatively high secure performance.
\end{lemma}

\begin{proof}
   As mentioned in Eq.~\eqref{eq: objecitve3}, the third optimization objective refers to the energy consumption of UAVs for flying from their original positions $\boldsymbol{P}^r$ to the optimized positions. Therefore, the space around $\boldsymbol{P}^r$ will achieve a lower third objective function value. Following this, the first and second objectives (shown in Eqs.~\eqref{eq: objecitve1} and~\eqref{eq: objecitve2}, respectively) are determined by the array factor of UVAAs, which is controlled by the relative positions of UAVs. Thus, even if UAVs fly around $\boldsymbol{P}^r$, they can still attain the optimal relative positions, thereby achieving a relatively high level of secure performance.
\end{proof}

\par Motivated by Lemma \ref{lemma:energy}, we propose a random walk-based initialization method that employs the information of $\boldsymbol{P}^r$ and a random walk approach to design the initial state of $\boldsymbol{P}$ part of each candidate solution $\boldsymbol{X}$. \textit{\uline{First}}, we generate the $N$ vectors drifted from $\boldsymbol{P}^r$ by using random walk, which is given by
\begin{equation}\label{eq:random_walk}
\begin{aligned}
	&\boldsymbol{P}^*= [ \boldsymbol{P}^*_1, \boldsymbol{P}^*_2,...,\boldsymbol{P}^*_N ] \\&=\left[0, f_c\left(2 r\left(1\right)-1\right), f_c \left(2 r\left(2\right)-1\right), \ldots, f_c\left(2 r\left(N\right)-1\right)\right],
\end{aligned}
\end{equation}

\noindent where $f_c$ calculates the cumulative sum, and $r(n)$ is a stochastic function which is given by 
\begin{equation}
	\label{eq:sto}
	r(n)=\left\{\begin{array}{ll}
1 & \text { if } rand >0.5 \\
0 & \text { if } rand \leqslant 0.5
\end{array}\right.,
\end{equation}
\noindent where $rand$ is a random number generated with uniform distribution in the interval of $[0, 1]$. 

\par \uline{\textit{Second}}, we use these vectors and random number generator to initialize the $n$th candidate solution of the population $\boldsymbol{X}_{n}^{init}$ as follows:
\begin{equation}\label{eq:generate}
\begin{aligned}
	&\boldsymbol{X}_{n}^{init}= [ \boldsymbol{P}^*_n, \boldsymbol{\Omega}_{rand}, \boldsymbol{u}_{rand}],
\end{aligned}
\end{equation}
\noindent where $\boldsymbol{\Omega}_{rand}$ is a random vector, an element of which is a random number generated with uniform distribution in the interval of $[0, 1]$. Moreover, $\boldsymbol{u}_{rand}$ is a random integer vector with two elements, and each element can be calculated as $round\left(rand*N_{U}\right)$.

\par Using this method, we can generate more informative initial candidate solutions, thereby facilitating the following iterative improvements of the MOALO algorithm.

%subsubsection
%Random Walk-based Initialization
%
\vspace{+0.5mm}
\noindent \textbf{Sorting-based Population Evolution.} In this subsection, we propose a novel scheme to improve the evolution of the population. We begin by proving a relationship between the archive and population evolution. 

\begin{lemma}
\label{lemma:trade-off}
    The distribution of solutions in the archive determines the search direction of the algorithm. In this case, MOALO may waste computational resources on excessively biased trade-offs, where one objective is prioritized while disregarding the other two objectives.
\end{lemma}

\begin{proof}
    In Step 4) of MOALO, the solutions in the current population are derived from both the solutions in the archive and the solutions from the previous population. Due to the crowding mechanism outlined in Step 3), MOALO tends to aim for a diverse range of trade-offs between objectives. Consequently, the archive may contain solutions that exhibit highly biased trade-offs. These biased solutions can misguide the search process of MOALO in subsequent iterations.
\end{proof}

\par Considering Lemma~\ref{lemma:trade-off}, we aim to eliminate less-efficient trade-offs by filtering candidate solutions of the archive. Our sorting-based population evolution method is as follows. 

\par \uline{\textit{First}}, in each iteration, we rank the populations separately according to the three optimization objectives. \uline{\textit{Second}}, we record the minimum value of the first and second optimization objectives and the maximum value of the third optimization objective, which are denoted as $f_{1_{min}}^{t}$, $f_{2_{min}}^{t}$, and $f_{3_{max}}^{t}$, respectively. \uline{\textit{Finally}}, we remove the candidate solutions below the set three thresholds. The three thresholds $\zeta_1^t$, $\zeta_2^t$, and $\zeta_3^t$ are given by
\begin{equation}\label{eq:zeta}
	\begin{aligned}
		\zeta_1 = f_{1_{min}}^{t} \times {\delta_1}, \quad
		\zeta_2 = f_{2_{min}}^{t} \times {\delta_2}, \quad
		\zeta_3 = f_{3_{max}}^{t} \times {\delta_3},
	\end{aligned}
\end{equation}
\noindent where $\delta_1$, $\delta_2$, and $\delta_3$ are three parameters which are ranged from 0 to 1. If we require high coverage for a specific objective, we can set the corresponding parameters relatively larger. Following this, the main steps of the sorting-based population evolution method are shown in Algorithm~\ref{algo:Sorting}. 
% By using this method, we can control the search direction of MOALO for boosting the optimization of objectives 1 and 2 while ensuring a reasonable increase of objective 3. As can be seen, we eliminate the meaningless candidate solutions with low values of the first and second objectives and relax the third objective in a low step through the iterations.

\begin{algorithm}[tb]
  \normalem
  \caption{Sorting-based Population Evolution}\label{algo:Sorting}
  \KwIn{$\mathcal{A}$, $\mathcal{P}$, $N$, current iteration $t$.}
  \KwOut{$\mathcal{A}$}
  $\mathcal{A} \leftarrow \{\mathcal{A}, \mathcal{P}\}$; (Denote the size of $\mathcal{A}$ as $N_A$ and the objectives of $a$th solution in $\mathcal{A}$ as [$f_{1,a}$,$f_{2,a}$,$f_{3,a}$])\\
  Rank $\mathcal{A}$ according to the first, second, and third objective values and record $f_{1_{min}}^{t}$, $f_{2_{min}}^{t}$, and $f_{3_{max}}^{t}$, respectively;\\

  Calculate $\zeta_1$, $\zeta_2$, and $\zeta_3$ by using Eq.~\eqref{eq:zeta};\\

  The dominated solutions are removed from $\mathcal{A}$;\\

  \For{$a=1$ to $N_A$}
  {
    \If {$f_{1,a}> \zeta_1$ and $mod(t,3)=0$ }
    {
    	The $a$th solution is removed from $\mathcal{A}$;\\
    }

    \If {$f_{2,a}> \zeta_2$ and $mod(t,3)=1$ }
    {
      The $a$th solution is removed from $\mathcal{A}$;\\
    }

    \If {$f_{3,a}> \zeta_3$ and $mod(t,3)=2$ }
    {
      The $a$th solution is removed from $\mathcal{A}$;\\
    }

  }
  Return $\mathcal{A}$;\\
\end{algorithm}

%ALGO
%INSGA-III
%
\begin{algorithm}[b]
  \normalem
  \caption{MOALO-RSI}\label{algo:MOALO-RSI}
  \KwIn{$\mathcal{P}$, $\mathcal{A}$, $N$, $t_{max}$ (maximum iteration).}
  \KwOut{$\mathcal{A}$}

  $\mathcal{P} \leftarrow \varnothing$, $\mathcal{A} \leftarrow \varnothing$;\\

  \For{$n=1$ to $N$}
  {
    Generate the solution $\boldsymbol{X}_{n}^{init}$ by using Eq. \eqref{eq:generate}; \\
    $\mathcal{P} \leftarrow \mathcal{P} \cup \left\{\boldsymbol{X}_{n}^{init}\right\}$;\\
  }
  \While{$t<t_{max}$}
  {
  	\For{$n=1$ to $N$}
	  {
	    Evaluate the objective values of $\boldsymbol{X}_{n}$ via Eqs.~\eqref{eq:AF}-\eqref{eq:formulation}, which are denoted as $F_{n}=[f_{1,n}, f_{2,n}, f_{3,n}]$;\\
	  }

	  Update $\mathcal{A}$ by using \textit{\textbf{Algorithm}} \ref{algo:Sorting};\\

  	  \For{$n=1$ to $N$}
	  {
	    Update $\boldsymbol{X}_{n}$ by using Eqs.~\eqref{eq:solution-update} and~\eqref{eq:integer-update};\\
	  }
    $t=t+1$;\\
  }
  Return $\mathcal{A}$;\\
\end{algorithm}

%subsubsection
%Random Walk-based Initialization
%
\vspace{+0.5mm}
\noindent \textbf{Integer Update Method.} Another critical issue of MOALO for solving the formulated optimization problem lies in dealing with the integer decision variables. Specifically, the update method mentioned in Eq. \eqref{eq:solution-update} can update only the candidate solution with continuous decision variables. However, the formulated optimization problem involves integer decision variables which is a challenging task for MOALO. Thus, we propose an integer update method in the following. 

\par Similar to Eq.~\eqref{eq:solution-update}, we employ the selected candidate solutions of the archive, \text{i.e.}, $\boldsymbol{X}^A$, to guide the update. Moreover, we also introduce the integer decision variables of the original candidate solution $\boldsymbol{X}_n$ and random integer generator to preserve inertia and increase randomness, respectively. Let $\boldsymbol{u}_{\mathcal{A}}= \boldsymbol{X}^A (\boldsymbol{u})$, $\boldsymbol{u}_{o}=\boldsymbol{X}_n (\boldsymbol{u})$, and $\boldsymbol{u}_{rand}=[randi(N_U), randi(N_U)]$ denote the integer decision variables of selected archive solution and original solution, and randomly generated integer, respectively, in which $randi(N)$ is a function that generates a random integer no more than $N$. Then, the integer decision variables are updated as follows:
\begin{equation}
	\label{eq:integer-update}
	\boldsymbol{u}=\left\{\begin{array}{ll}
\boldsymbol{u}_{\mathcal{A}}, & rand<\frac{1}{3} \\
\boldsymbol{u}_{o}, & \frac{2}{3}>rand>\frac{1}{3} \\
\boldsymbol{u}_{rand}, & \text{otherwise} \\
\end{array}\right..
\end{equation}

\noindent As can be seen, the updated integer decision variables are guided by elite, inertia, and randomness mechanisms~\cite{Mirjalili2017b}, thereby achieving a more balanced search of the integer solution space.

%subsubsection
%Main steps and computational complexity
%
\vspace{+0.5mm}
\noindent \textbf{Main Steps and Computational Complexity.} By introducing the aforementioned random walk-based initialization, sorting-based population evolution, and integer update methods, we summarize a novel MOALO-RSI to solve the formulated optimization problem. Let $N$ and $t_{max}$ be the population size and maximum iteration of the algorithm, the main steps of MOALO-RSI are shown in Algorithm~\ref{algo:MOALO-RSI}. Note that $t_{max}$ can be determined by historical experience. Besides, in computational resource-sensitive scenarios, it is also possible to establish thresholds for the three objectives. In such cases, iterations can be terminated promptly once a solution surpassing the predefined threshold is obtained. The computation complexity of our MOALO-RSI is analyzed as follows. 

%FIGURE
%path
%
\begin{figure*}
    \centering
    \subfloat[Antenna gain of UVAA1]{
       \includegraphics[width=0.22\linewidth]{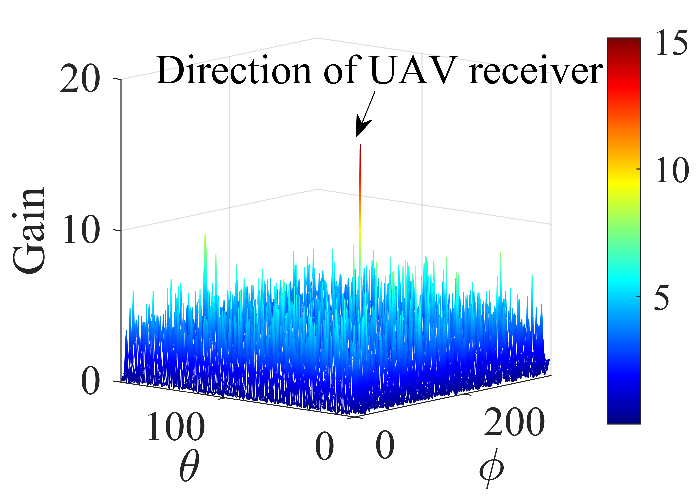}\label{fig:vr-gain1}}
    % \hfill
    \subfloat[Antenna gain of UVAA2]{
       \includegraphics[width=0.22\linewidth]{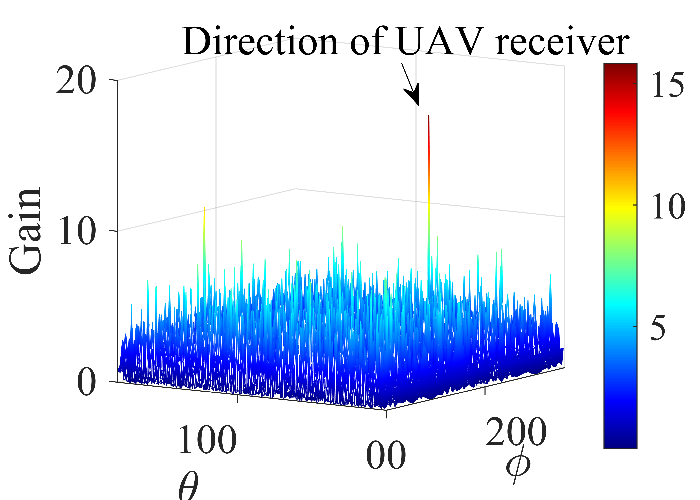}\label{fig:vr-gain2}}
    % \hfill
     \subfloat[Position changes of UVAA1 and UVAA2]{
       \includegraphics[width=0.52\linewidth]{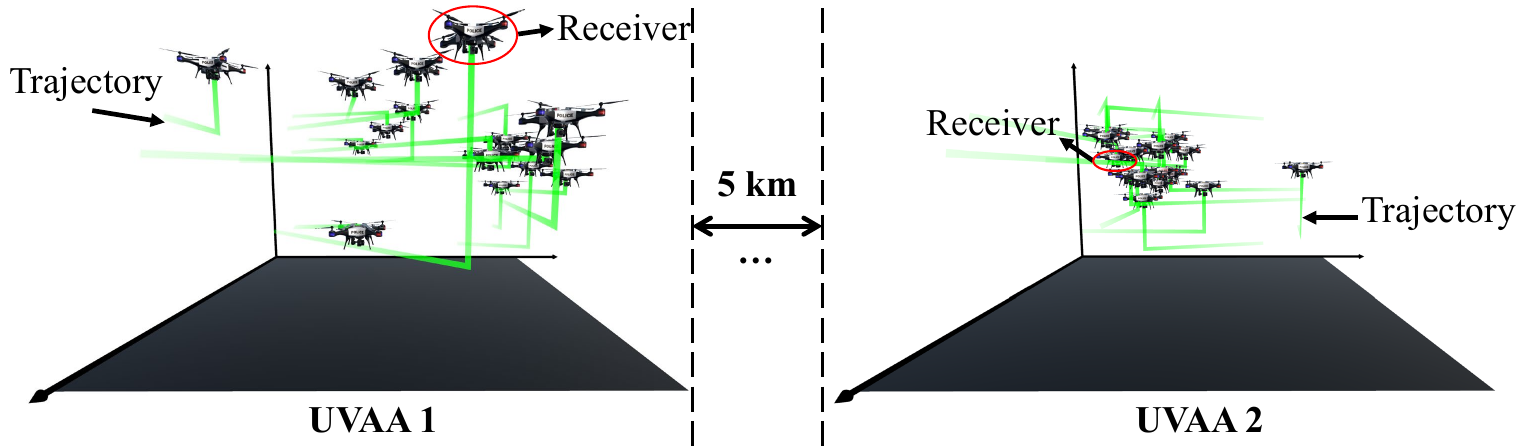}\label{fig:vr-path1}} 
     % \hfill
  \caption{Visualization results obtained by our proposed MOALO-RSI. }
  \label{fig:Visualization-results}
\end{figure*}

\vspace{+0.5 mm}
\begin{lemma}
	\textit{The complexity of MOALO-RSI is $\mathcal{O}(N_o N^2)$.}
\end{lemma}

\vspace{+0.5 mm}
\begin{proof}
	The computations of the objective functions and the crowding mechanism determine the computational complexity of the multi-objective optimization algorithms. If the number of optimization objectives is $N_o$, then the computational complexity for calculating the objective functions and crowding mechanism computation are $\mathcal{O}(N_o N)$ and $\mathcal{O}(N_o N_{Arc} \log N_{Arc})$, respectively. In most cases, the size of the Pareto archive is the same as the population $N$, which means that the computational complexity for the non-dominated sorting is $\mathcal{O}(N_o N^2)$. Therefore, the complexities of the MOALO and MOALO-RSI are both $\mathcal{O}(N_o N^2)$. 
\end{proof}

\par As can be seen, the computational complexity of the proposed MOALO-RSI does not increase after being improved.

%subsubsection
%Deployment Scheme
%
\vspace{+0.5mm}

\noindent \textbf{Deployment Scheme.} MOALO-RSI can be easily extended to a parallel distributed version by using synchronism homogeneity island model~\cite{Li2023} (which will be considered in our future work). Thus, we can execute MOALO-RSI on an accessible centralized high-performance computing device (\textit{e.g.}, mobile terrestrial workstations) or run its parallel distributed version by using the computation device within the UAV swarm. In this case, the communications between UAV swarms can be done by non-optimized robust DCB, while the communications within the same swarm can be accomplished by the inner-swarm communication protocol~\cite{Feng2013}. 
% \noindent \textbf{Deployment Scheme.} When our proposed MOALO-RSI is deployed in practice, it necessitates a scheme for acquiring prior information and sharing results. We discuss a viable scheme as follows. 

% \begin{enumerate}
%     \item UAV swarms use non-optimized robust DCB to share crucial details such as their positions, locations of identified eavesdroppers, and other prior information. Note that non-optimized DCB enables long-range transmission but lacks secure measures. Therefore, the information is encrypted to ensure confidentiality.

%     \item Running our MOALO-RSI to get the optimization results. Note that MOALO-RSI can be easily extended to a parallel distributed version by using synchronism homogeneity island model~\cite{Li2023} (which will be considered in our future work). Thus, we can execute MOALO-RSI on an accessible centralized high-performance computing device (\textit{e.g.}, mobile terrestrial workstations) or run its parallel distributed version by using the computation device within the UAV swarm. In the latter case, UAVs can use the intra-swarm communication protocol~\cite{Feng2013} to achieve information exchange.

%     \item The UAV swarm which runs MOALO-RSI sends the optimization results and schedules to another by using non-optimized robust DCB. Then, two UAV swarms will perform the DCB-based secure communications according to the optimization results of MOALO-RSI.
% \end{enumerate}

%
%Simulation results and analysis
%
\section{Simulation Results} % (fold)
\label{sec:simulation_results_and_analysis}

\par In this section, we provide important simulation results to evaluate the effectiveness of the proposed method.

\begin{table}
	\centering
	\caption{Numeral results obtained by baselines and our MOALO-RSI.}
	\label{tab:result-plos-small}
	\begin{tabular}{lllll}
		\toprule
		\textbf{Method} 	& \bf{$f_{1}$ [bps] }          & \bf{$f_{2}$ [dB] } & \bf{$f_{3}$ [J] } \\ \midrule
    LAA-Swarm               & $7.59 \times 10^{5}$    & $-0.21$          & $8.23 \times 10^{4}$  \\
    MOGOA~\cite{Mirjalili2018}                   & $1.90 \times 10^{6}$    & $-0.92$          & $1.04 \times 10^{5}$    \\
    MOMVO~\cite{Mirjalili2017}                    & $2.00 \times 10^{6}$    & $-1.73$          & $1.11 \times 10^{5}$   \\
    MSSA~\cite{Mirjalili2017a}                     & $1.93 \times 10^{6}$    & $-0.48$          & $8.08 \times 10^{4}$   \\
    MODA~\cite{Mirjalili2016}                     & $1.98 \times 10^{6}$    & $-2.32$          & $9.41 \times 10^{4}$  \\
    MOALO~\cite{Mirjalili2017b}       & $2.01 \times 10^{6}$    & $-1.63$          & $1.04 \times 10^{5}$   \\
    Our MOALO-RSI               & \bm{$2.14 \times 10^{6}$}    & \bm{$-2.48$}      & \bm{$7.14 \times 10^{4}$} \\\bottomrule
	\end{tabular}
\end{table}

%subsection
%Simulation Setups
%
\vspace{+0.5mm}
\noindent \textbf{Simulation Setups.} In the simulations, the numbers of UAVs of two UAV swarms are set as 16, and the UAV heights vary from 70 m to 120 m. Moreover, these two UAV swarms are separately distributed in two 100 m $\times$ 100 m areas, and the distance between these two areas is about 5 km. The collision distance $d_{min}$ between two arbitrary UAVs is set as 0.5 m. Other parameters related to communications and the UAV energy model follow~\cite{Mozaffari2019} and~\cite{Li2023TMC}, respectively. For comparison, we adopt the following baselines:  

\begin{itemize}
    \item \textit{LAA-Swarm:} Two UAV swarms separately form two linear antenna arrays and randomly select a receiver from a different UAV swarm. 

    \item \textit{State-of-the-art Baseline Algorithms:} We consider state-of-the-art multi-objective swarm intelligence algorithms, including multi-objective grasshopper optimization algorithm (MOGOA)~\cite{Mirjalili2018}, multi-objective multi-verse optimizer (MOMVO)~\cite{Mirjalili2017}, multi-objective salp swarm algorithm (MSSA)~\cite{Mirjalili2017a}, multi-objective dragonfly algorithm (MODA)~\cite{Mirjalili2016}, and MOALO~\cite{Mirjalili2017b}. Note that these algorithms employ the proposed integer update method to handle the integer decision variables of the formulated problem. Their population size and maximum iteration are 50 and 300, respectively, while key parameters follow the initial settings of their source papers. 
\end{itemize}

% Figure
% PS
%
\begin{figure}
  \centering
  \includegraphics[width=3 in]{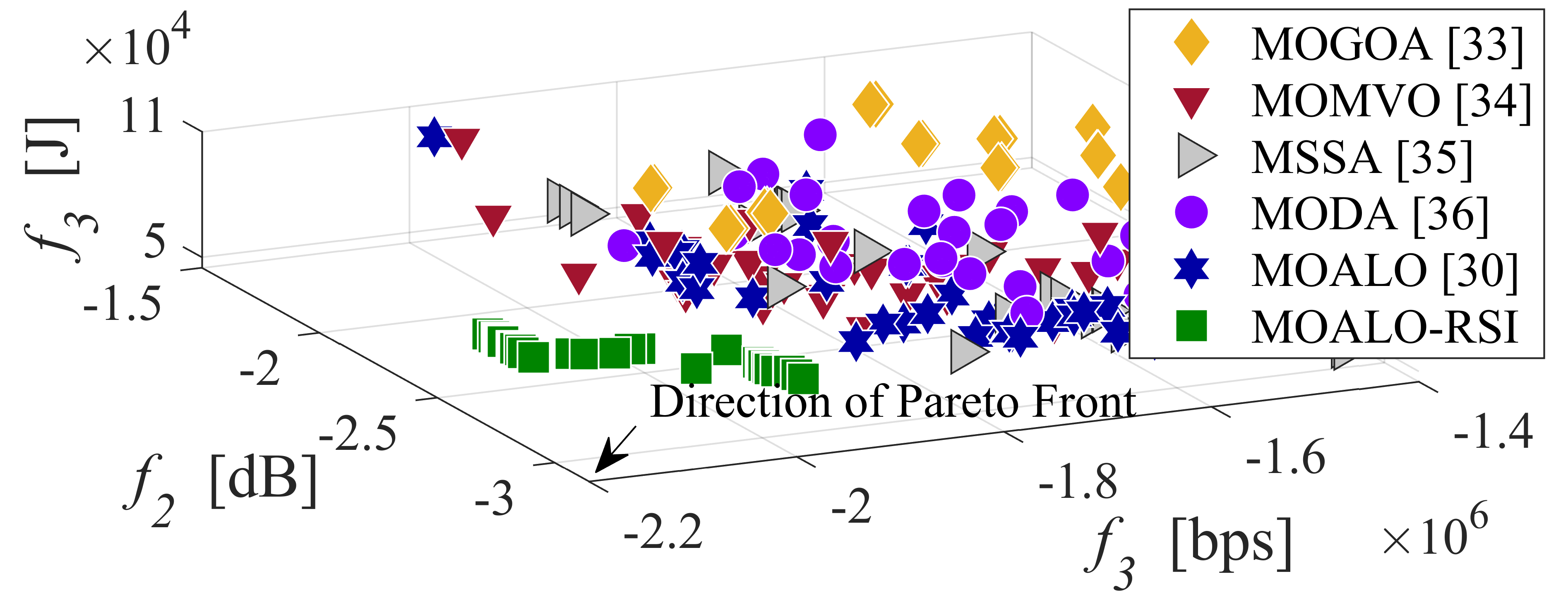}
  \caption{Pareto solutions obtained by benchmarks and the proposed MOALO-RSI. All the nodes denote the objective values of the candidate solutions obtained by different algorithms.}
  \label{fig:PS}
\end{figure}

%FIGURE
%Robustness verification results
%
\begin{figure*}
    \centering
    \subfloat[Phase synchronization error]{
       \includegraphics[width=0.22\linewidth]{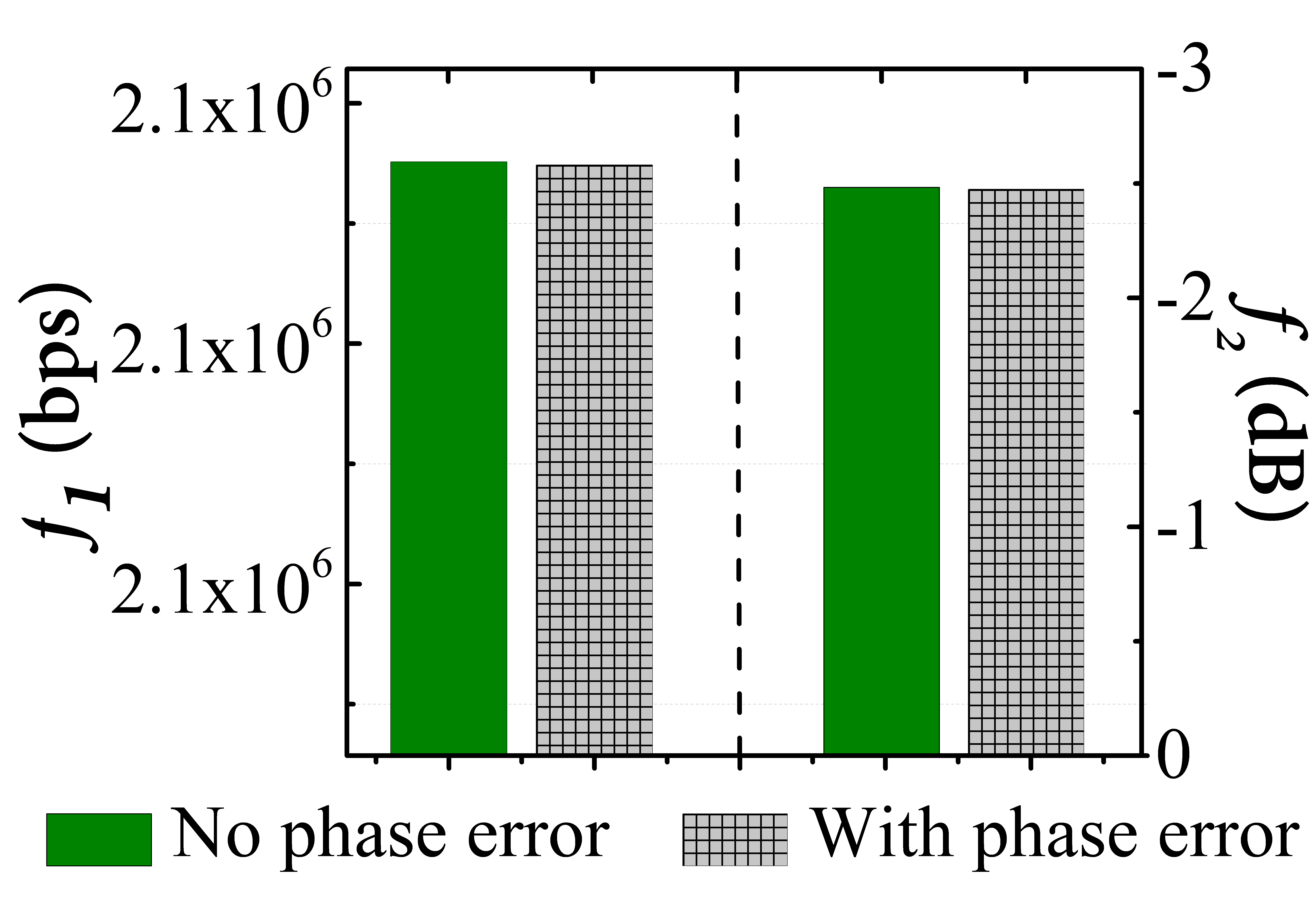}\label{fig:rv-phase}}
    \subfloat[CSI errors]{
       \includegraphics[width=0.37\linewidth]{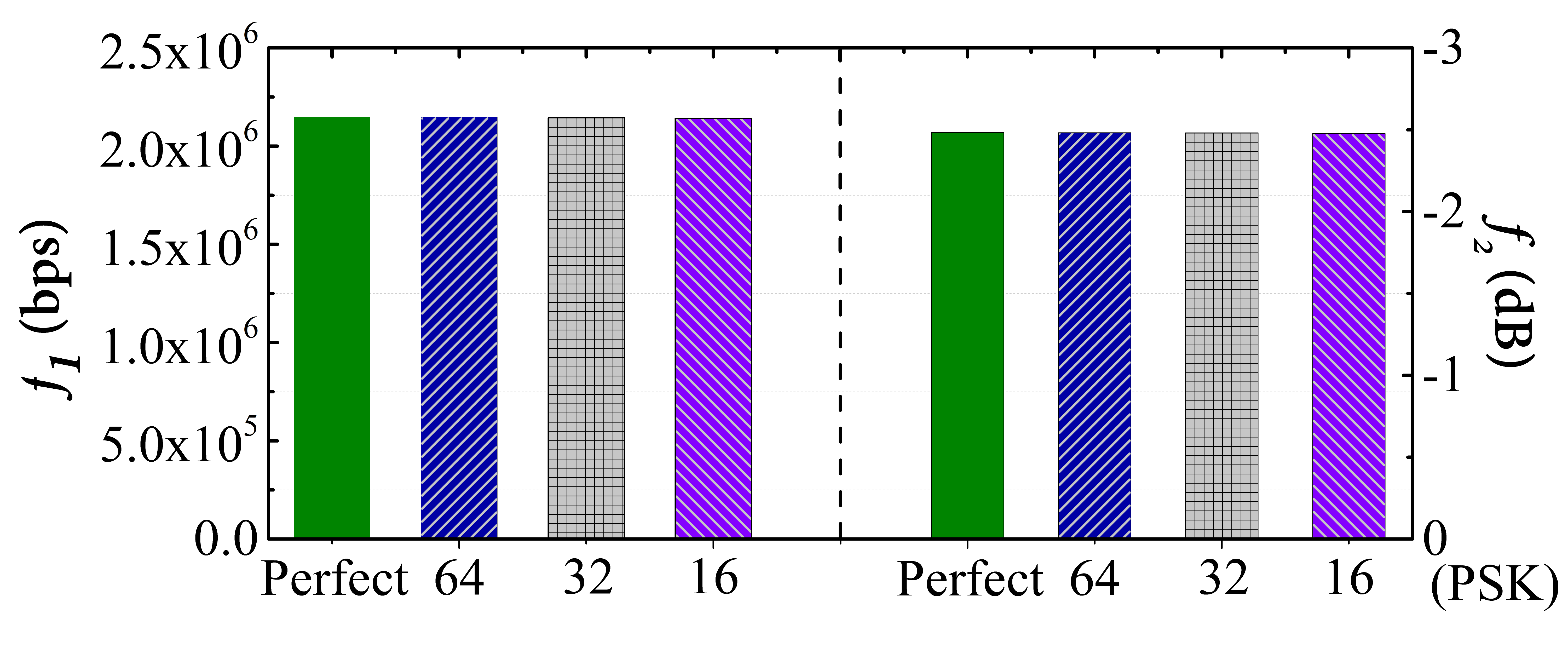}\label{fig:rv-CSI}} 
     \subfloat[UAV jitter]{
       \includegraphics[width=0.37\linewidth]{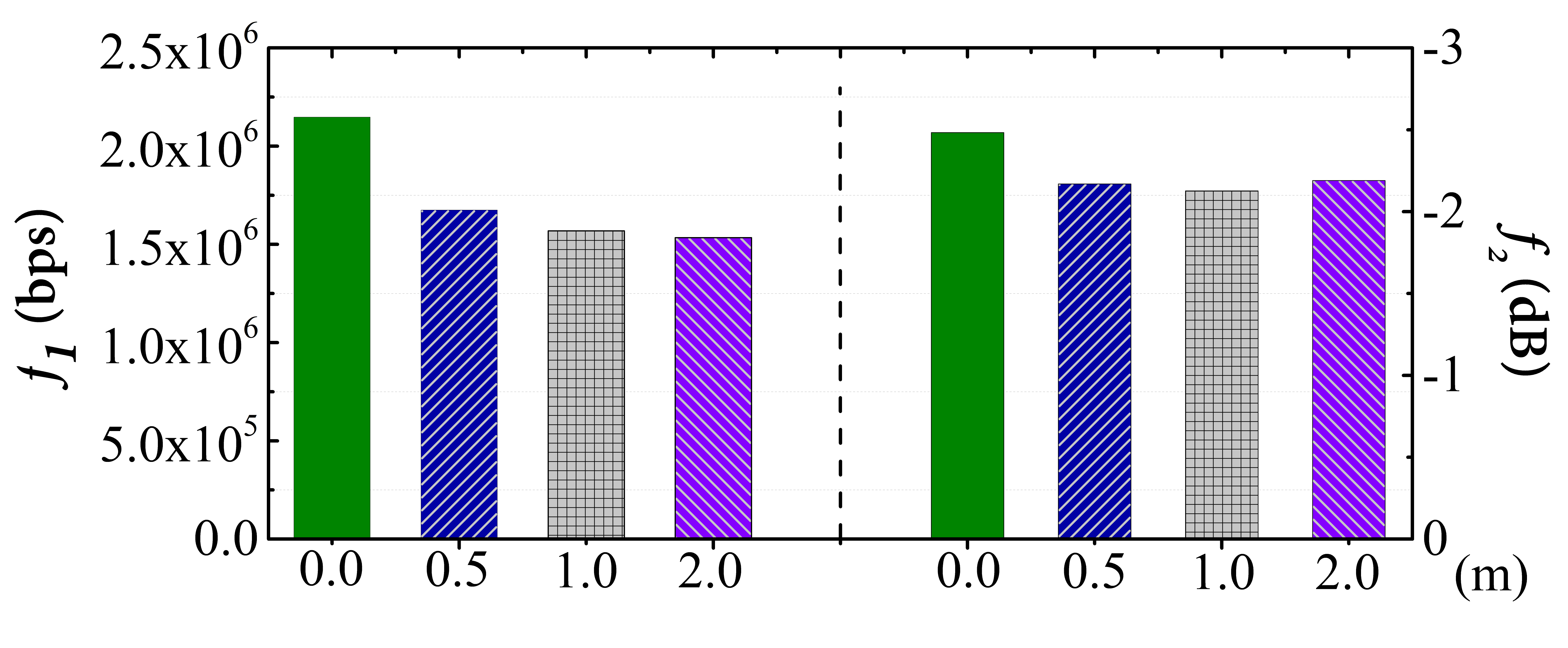}\label{fig:rv-position}}
       \\
  \caption{Robustness verification results (in terms of $f_1$ and $f_2$) under some special cases.}
  \label{fig:Robustness-results}
\end{figure*}

%subsection
%Visualization Results.
%
\vspace{+0.5mm}
\noindent \textbf{Visualization Results.} For ease of presentation, we employ Unity 3D and Matlab to visualize the results and demonstrate the effectiveness of the solution obtained by our MOALO-RSI. First, we present the antenna gains of the two UVAAs in Figs.~\ref{fig:Visualization-results}\subref{fig:vr-gain1} and~\ref{fig:Visualization-results}\subref{fig:vr-gain2}. It is evident that, except for the target directions, the antenna gains of both UVAAs are relatively low. This indicates that two-way aerial communications can achieve high secrecy capacities. Second, in Fig.~\ref{fig:Visualization-results}\subref{fig:vr-path1}, we illustrate the trajectories of the UAVs in the two UAV swarms during constructing UVAAs. Notably, the UAVs exhibit minimal position changes, resulting in energy savings. Thus, these visualization results are strong evidence that our solution achieves excellent security performance and energy efficiency.

% As can be seen, our method does not degrade significantly in these cases, which shows a certain degree of robustness.

% Figure
% flight_control
%
\begin{figure}
  \centering
  \includegraphics[width=3 in]{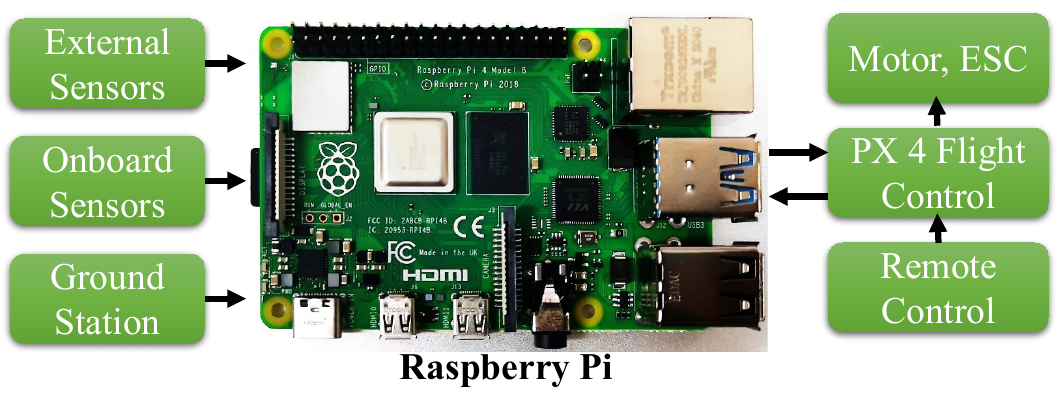}
  \caption{An example of autonomous UAV systems based on Raspberry Pi 4B. }
  \label{fig:flight_control}
\end{figure}

%subsection
%Comparisons and Analyses
%
\vspace{+0.5mm}
\noindent \textbf{Comparisons and Analyses.} In this part, MOALO-RSI is compared with other baseline methods mentioned in simulation setups in solving our optimization problem. Table \ref{tab:result-plos-small} shows the numeral results from our MOALO-RSI and the comparison benchmarks in terms of the three objectives as given in Eqs.~\eqref{eq: objecitve1},~\eqref{eq: objecitve2} and~\eqref{eq: objecitve3} (\textit{i.e.}, $f_1$, $f_2$, and $f_3$). For ease of analysis, we assume $f_1$ is the most crucial objective of this work, and select the solution with the best $f_1$ value from the Pareto solution sets as the final solution by using the automatic method in~\cite{Ferreira2007}. As can be seen, MOALO-RSI is superior to the LAA-Swarm method which is most likely to be employed in practice, indicating that the considered optimization approach is non-trivial. Moreover, MOALO-RSI outmatches various multi-objective swarm intelligence algorithms, implying that it is more suitable for solving the formulated problem. Thus, the formulation and the proposed enhanced measures of the MOALO-RSI are valid and effective, which provides an insight that using the problem's physical properties (\textit{e.g.}, the properties summarized in Lemmas~\ref{lemma:energy} and~\ref{lemma:trade-off}) to enhance the swarm intelligence algorithm.

\par Then, Fig. \ref{fig:PS} shows the Pareto solutions (\textit{i.e.}, candidate solutions in archive) procured by the proposed MOALO-RSI and other baseline algorithms. As can be seen, the candidate solutions obtained by MOALO-RSI are approaching the ideal Pareto front direction, and mostly dominate the candidate solutions obtained by other baseline algorithms. Moreover, our MOALO-RSI avoids obtaining excessively biased trade-offs (\textit{e.g.}, the trade-off with -0.02 dB SLL obtained by MOALO). The reason is that the proposed sorting-based population evolution method remote excessively biased trade-offs in each iteration, thereby facilitating reasonable trade-off distributions.

%subsection
%Robustness Verification.
%
\vspace{+0.5mm}
\noindent \textbf{Robustness Verification.} In this part, we evaluate the robustness of our proposed method. \uline{\textit{First}}, we simulate the phase synchronization error mentioned in~\cite{Ahmad2022}, which follows a Gaussian distribution with zero-mean and variance $\zeta^{2}=\omega_{c}^{2} q_{1}^{2} \Delta T+\omega_{c}^{2} q_{2}^{2} \Delta T^{3}/3$ (these parameters follow~\cite{Ahmad2022}). Fig.~\ref{fig:Robustness-results}\subref{fig:rv-phase} demonstrates that our solution exhibits negligible performance loss under phase synchronization errors. \uline{\textit{Second}}, we introduce varying degrees of CSI errors which are induced when using different length CSI codebooks from~\cite{Ahmad2022}, including errors of 16-PSK, 32-PSK, and 64-PSK codebooks. Note that a longer codebook tends to yield smaller CSI errors. As depicted in Fig.~\ref{fig:Robustness-results}\subref{fig:rv-CSI}, the performance loss in terms of $f_1$ and $f_2$ is generally insignificant across most scenarios, particularly when employing codebooks longer than 32-PSK. \uline{\textit{Finally}}, we examine four UAV jitter conditions, where the maximum drifts are set to 0.5 m, 1 m, and 2 m, respectively. As observed from Fig.~\ref{fig:Robustness-results}\subref{fig:rv-position}, the performance gaps between the non-drift and position-drifted cases are minimal for $f_2$. However, for $f_1$, there is a slight degradation when drifts are present, although it remains acceptable for relatively small drifts. Overall, our proposed method demonstrates a certain degree of robustness.
% under various special cases, including phase synchronization error, CSI error, and UAV jitter.

%subsection
%Practicality Analysis.
%
\vspace{+0.5mm}
\noindent \textbf{Practicality Analysis.} We assess the practicality of our method and encryption/decryption techniques. As shown in Fig.~\ref{fig:flight_control}, we utilize the Raspberry Pi 4B as the UAVs' flight control system, as it is a common setup in widely used UAV platforms (\textit{e.g.}, PX4 autopilot) and previous studies (\textit{e.g.},~\cite{Zhou2021}). As discussed in Section \ref{sec:algorithm}, we assume that one UAV swarm runs a parallel distributed version of our MOALO-RSI, and omit the step for calculating optimization objective values as it is often substituted with proxy models~\cite{Jeong2005} in real-world scenarios. Additionally, three common encryption/decryption methods, namely, data encryption standard (DES), advanced encryption standard (AES), and Rivest-Shamir-Adleman (RSA), are introduced for further comparisons~\cite{Bhanot2015}.

Experimental results demonstrate that a one-time calculation of our MOALO-RSI can be completed within 40.18 s. Furthermore, the calculation times for encrypting and decrypting 200 MB of data using DES, AES, and RSA are 12.07 s, 9.29 s, and 1567.59 s, respectively. As such, when the data volume exceeds approximately 1 GB, our proposed method achieves obvious advantages in computing time. This is because the encryption/decryption techniques need to continuously calculate over time while our method only needs one-time calculation. Therefore, the advantages of our method become more apparent as the amount of transferred data increases.
% Overall, our method can be deployed in low computing power platforms of UAVs and is beneficial to saving computational resources, especially when the transmitted data is large-scale.

%
%Conclusion
%
\section{Conclusion} % (fold)
\label{sec:conclusion}

\par This paper investigated a DCB-enabled aerial two-way communication of two UAV swarms under eavesdropper collusion. To achieve secure and energy-efficient communications, we formulated an MOP to improve the minimum two-way known secrecy capacity, minimize the maximum SLL, and restrict the energy consumption of the UAVs simultaneously. Due to the NP-hardness of the MOP, we proposed an enhanced multi-objective swarm intelligence method, \textit{i.e.}, MOALO-RSI, to solve the problem. Simulation results demonstrated that MOALO-RSI outmatched MOALO and other baseline algorithms and was robust. Finally, experimental results showed that MOALO-RSI can run on limited computing power platforms and can save computational resources.
% section conclusion (end)
% Specifically, we considered the worst wiretap case that the eavesdroppers adopt MRC based on signal detection. Moreover, our proposed method is with the robustness to mitigate the impacts of phase synchronization error, CSI error, and UAV jitter. 

\section*{Acknowledgement}
\par This study is supported in part by the National Natural Science Foundation of China (62172186, 62002133, 61872158, 62272194), and in part by the Science and Technology Development Plan Project of Jilin Province (20230201087GX).

\normalem
\bibliographystyle{IEEEtran}
\bibliography{myref}

\end{document}